\documentclass[a4paper,english,11pt]{article}

\usepackage[german, english]{babel}
\usepackage[pdftex]{color, graphicx}
\usepackage{dsfont}
\usepackage{amsmath}
\usepackage{amssymb}
\usepackage{cancel}
\usepackage[utf8]{inputenc}
\usepackage{enumerate}
\usepackage{ifthen}
\usepackage{paralist}
\usepackage{fancyhdr}

\newtheorem{theorem}{Theorem}[section]
\newtheorem{corollary}[theorem]{Corollary}
\newtheorem{proposition}[theorem]{Proposition}
\newtheorem{lemma}[theorem]{Lemma}
\newtheorem{remark}[theorem]{Remark}
\newtheorem{definition}[theorem]{Definition}
\newtheorem{claim}[theorem]{Claim}
\newtheorem{problem}[theorem]{Problem}
\newtheorem{notation}[theorem]{Notation}
\newtheorem{beispiel}[theorem]{Example}

\newtheorem{conjecture}[theorem]{Conjecture}

\newenvironment{theo}[1][\empty]{\begin{theorem} 
\ifthenelse{\equal{#1}{\empty}} {}  {\itshape(#1)} \normalfont ~\\}{\end{theorem}}
\newenvironment{cor}[1][\empty]{\begin{corollary} 
\ifthenelse{\equal{#1}{\empty}} {}  {\itshape(#1)} \normalfont ~\\}{\end{corollary}}
\newenvironment{prop}[1][\empty]{\begin{proposition} 
\ifthenelse{\equal{#1}{\empty}} {}  {\itshape(#1)} \normalfont ~\\}{\end{proposition}}
\newenvironment{lem}[1][\empty]{\begin{lemma} 
\ifthenelse{\equal{#1}{\empty}} {}  {\itshape(#1)} \normalfont ~\\}{\end{lemma}}
\newenvironment{rem}[1][\empty]{\begin{remark} 
\ifthenelse{\equal{#1}{\empty}} {}  {\itshape(#1)} \normalfont ~\\}{\end{remark}}
\newenvironment{defi}[1][\empty]{\begin{definition} 
\ifthenelse{\equal{#1}{\empty}} {}  {\itshape(#1)} \normalfont ~\\}{\end{definition}}

\newenvironment{prb}[1][\empty]{\begin{problem} 
\ifthenelse{\equal{#1}{\empty}} {}  {\itshape(#1)} \normalfont ~\\}{\end{problem}}

\newenvironment{proof}[1][\empty]{\ifthenelse{\equal{#1}{\empty}} {\paragraph{\textbf{Proof.}}~\\}  
{\paragraph{\textbf{Proof} of #1.}~\\}} {\hfill $\Box$}

\newcommand {\N}[0] {\mathbb{N}}
\newcommand {\R}[0] {\mathbb{R}}
\newcommand {\B}[0] {\mathbb{B}}

\newcommand{\conv}[0] {\mathrm{conv}}
\newcommand{\aff}[0] {\mathrm{aff}}

\newcommand{\vol}[0]{\mathrm{vol}}

\newcommand{\bd}[0]{\mathrm{bd}}
\renewcommand{\int}[0]{\mathrm{int}}

\newcommand{\CC}[0]{\mathcal{C}}

\newcommand{\CL}[0]{\mathcal{L}}

\begin{document}
\begin{titlepage}
  \renewcommand{\thefootnote} {\fnsymbol{footnote}} \vspace*{1,0cm}
  \begin{center}
 \LARGE  {\bf No Dimension Independent Core-Sets for Containment under Homothetics} 
 \\[25mm]
 \normalsize
 \begin{minipage}[c]{10pt}
   \mbox{}
 \end{minipage}
 \begin{minipage}[c]{170pt}
   Ren\'{e} Brandenberg \\
   ~\\
   Zentrum Mathematik\\
   Technische Universität München\\
   Boltzmannstr.~3\\
   85747 Garching bei München \\
   Germany \\ 
   ~\\
   E-mail: brandenb@ma.tum.de
 \end{minipage} 
 \hfill
 \begin{minipage}[c]{10pt}
   \mbox{}
 \end{minipage}
\begin{minipage}[c]{170pt}
  Stefan K\"{o}nig\footnotemark[1] \\
  ~\\
  Zentrum Mathematik\\
  Technische Universität München\\
  Boltzmannstr.~3\\
  85747 Garching bei München\\
  Germany \\
  ~\\
  E-mail: koenig@ma.tum.de
 \end{minipage}
 
 \vspace*{2.5cm}

\begin{minipage}[c]{410pt}
\noindent {\sc Abstract.}    
This paper deals with the containment problem under homothetics which has the minimal enclosing ball (MEB) problem as a prominent representative. 
We connect the problem to results in classic convex geometry and introduce a new
series of radii, which we call core-radii. For the MEB problem, these radii have
already been considered from a different point of view and sharp inequalities between
them are known. In this paper sharp inequalities between core-radii for general
containment under homothetics are obtained.

Moreover, the presented inequalities are used to derive sharp upper bounds on
the size of core-sets for containment under homothetics. In the MEB case, this
yields  a tight (dimension independent) bound for the size of such core-sets. In
the general case, we show that there are core-sets of size linear in the
dimension and that this bound stays sharp even if the container is required to be symmetric. 
\end{minipage}

\vspace*{1cm}

\begin{minipage}[c]{410pt}
\end{minipage}
\end{center}

\footnotetext[1]{The second author gratefully acknowledges the support of the TUM Graduate School's Faculty Graduate Center \emph{International School of Applied Mathematics} at Technische Universität München, Germany.}

\renewcommand{\thefootnote}%
{\arabic{}}

\end{titlepage}

\selectlanguage{english}

\renewcommand{\subseteq}{\subset}
\renewcommand{\supseteq}{\supset}
\newcommand{\MCP}{MCP$_{Hom}$}

\section{Introduction}

 Many well known problems in computational geometry can be classified as some type of optimal containment problem, 
 where the objective is to find an extremal representative $C^*$ of a given class of convex bodies, 
 such that $C^*$ contains a given point set $P$ (or vice versa). 
 These problems arise in many different applications, e.g.\ facility location, shape fitting and packing problems, 
 clustering, pattern recognition or statistical data reduction. Typical representatives are the minimal enclosing ball (MEB) problem, smallest enclosing cylinders, slabs, boxes, or ellipsoids; see \cite{gritzmannKlee-94} for a survey. Also the well known $k$-center problem, where $P$ is to be covered by $k$ homothetic copies of a given container $C$, has to be mentioned in this context.

 Because of its simple description and the multitude of both theoretical and practical applications 
 there is vast literature concerning the MEB problem. In recent years, a main focus has been on so called core-sets, 
 i.e. small subsets $S$ of $P$ requiring balls of (almost) the same radius  
to be enclosed as $P$ itself. For the Euclidean MEB problem algorithms constructing core-sets of sizes only depending on the approximation quality 
 but neither on the number of points to be enclosed nor the dimension have been developed in 
 \cite{badoiu, bhpi-02, clarkson, yapv-07}. This yields not only another fully polynomial time approximation scheme (FPTAS) 
 for MEB, but also a polynomial time approximation scheme (PTAS) for the harder Euclidean $k$-center problem which also works very well in practice \cite{br-07}.

 However, all variants of core-set algorithms for MEB are based on the so called half-space 
 lemma \cite{badoiu,halbraumlemma} or equivalent optimality conditions, a property characterizing the 
 Euclidean ball \cite{gritzmannKlee-92}, thus not allowing immediate
 generalization to the 
 superordinate \emph{Containment under Homothetics} that we consider in this paper:

\begin{prb} \label{prb:contHomo}%
For $P \subseteq\R^d$ compact and  $C \subseteq \R^d$ a full-dimensional compact convex set (called container) the \emph{minimal containment problem under homothetics} (\MCP) is to find the least dilatation factor $\rho \geq 0$, such that a translate of $\rho C$ contains $P$. In other words:
 we are looking for a solution to the following optimization problem:
\begin{equation} \label{eq:problem}
\begin{array}{rl}
 \min &  \rho \\
 s.t.	&  P \subseteq c + \rho C \\
 			&	 c \in \R^d \\
 			&	 \rho \geq 0
\end{array}
\end{equation}
where  $c +\rho C := \{c + \rho x : x \in C\}$. 
The assumption that $C$ be full-dimensional ensures that Problem \ref{prb:contHomo} has a feasible solution for every $P$. Hence, a compactness argument \cite{gritzmannKlee-92} shows that the minimum in \eqref{eq:problem} is attained for every $P$ and $C$.  We write $R(P,C)$ for the optimal value of \eqref{eq:problem} and call it the \emph{$C$-radius} of $P$.
Hence, if $C$ is a Euclidean ball and $P$ is finite this specializes to the MEB problem. If $C$ is 0-symmetric this is the problem of computing the outer radius of $P$ with respect to the norm $\|\cdot\|_C$ induced by the gauge body $C$ as already considered e.g. in \cite{bonnesenFenchelT}. 
\end{prb}

\begin{figure}[h]
\centering
\includegraphics[width=0.45\textwidth]{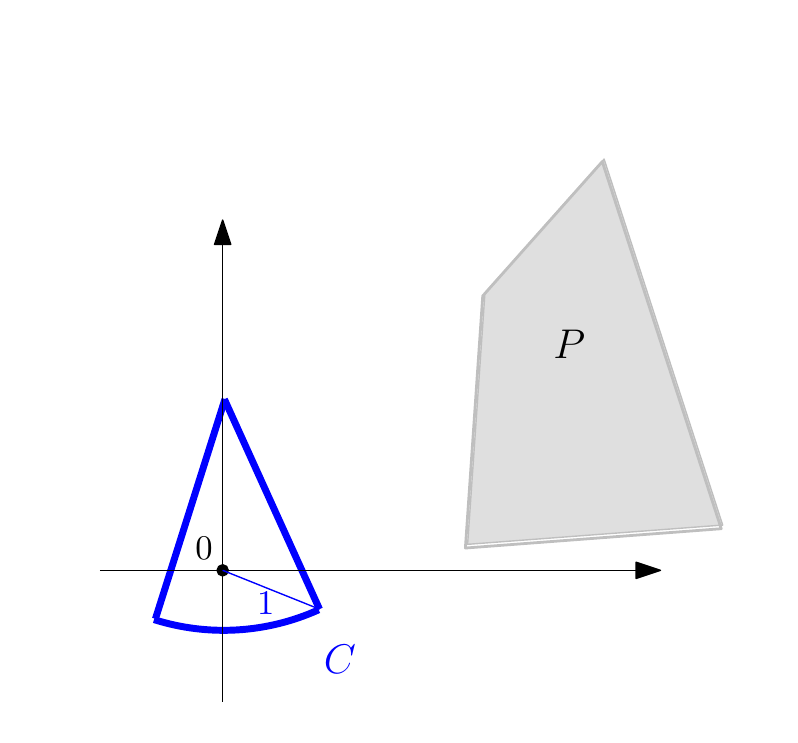}
\hfill
\includegraphics[width=0.45\textwidth]{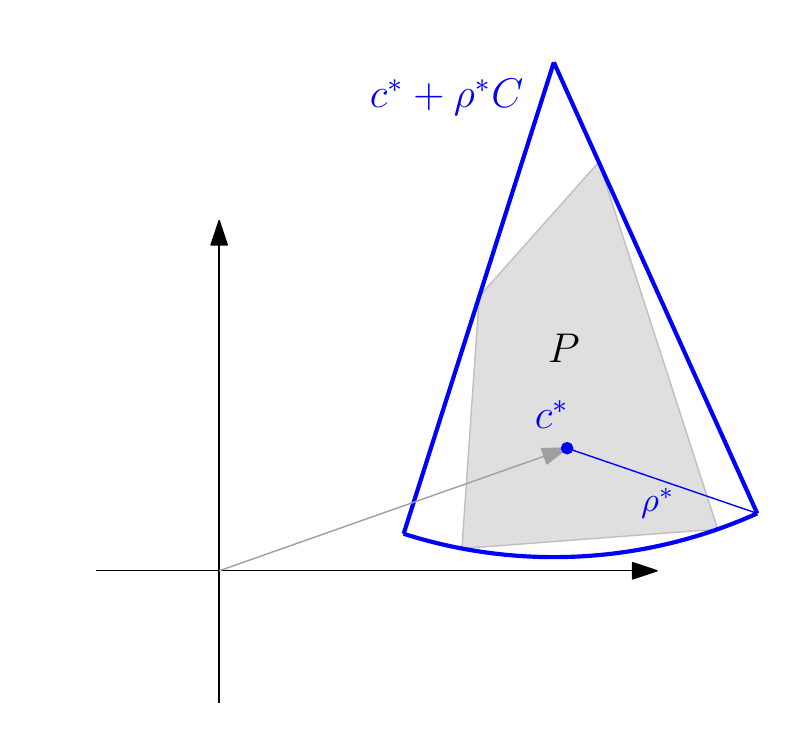}
\caption{Left: Possible input for problem \ref{prb:contHomo}. Right: An optimal solution for the input on the left. \label{fig:contHomo}}
\end{figure}

Besides direct applications Problem \ref{prb:contHomo} is often the basis for solving much harder containment problems (e.g. containment under similarities), which already gives a reason for an intensive search for good (approximation) algorithms. Compared to the approach in \cite{svz}, approximation via core-sets has the additional advantage that it may be turned into a PTAS for the $k$-center problem as demonstrated in \cite{bhpi-02}. Whereas there is a rich literature on the Euclidean MEB problem (and its core-sets) that exhibit many nice properties, only little is known about the general case and how much of the Euclidean properties carry over to Problem \ref{prb:contHomo}. (For an overview of possible solution strategies depending on given container classes, we refer to \cite{brandenbergRoth-10}.)

For $P$ and $C$ as in Problem \ref{prb:contHomo}, we call a subset $S \subset P$ an $\varepsilon$-core-set for some $\epsilon \geq 0$ if 
 $$R(P,C) \leq (1+\varepsilon)R(S,C).$$

Answering the questions for the size of core-sets for Problem \ref{prb:contHomo}, this paper proves the following result:

\begin{theo}[No sublinear core-sets for containment under homothetics] \label{theo:main}%
  For every  compact set $P \subset \R^d$, every container  $C\subseteq \R^d$, and $\varepsilon \ge 0$ there exists an 
  $\varepsilon$-core-set of $P$ of size at most 
  $\left \lceil \frac{d}{1+\varepsilon}\right\rceil+1$. Moreover, for any  $\varepsilon <1$ there exists a body $P  
  \subseteq\R^d$ and a 0-symmetric container $C$ such that no smaller subset of $P$ suffices.
\end{theo}

In order to prove the positive part of Theorem \ref{theo:main}, we will state several new geometric identities and inequalities between radii of convex sets, which connect Problem \ref{prb:contHomo} to results in classic convex geometry. The negative part of the theorem (i.e. that the bound cannot be improved even for 0-symmetric containers) then follows by proving that these inequalities (and so the resulting bounds on core-set sizes) are best possible.

Moreover, the connection between core-sets and a series of radii from convex geometry will enable us to give a sharp upper bound for the size of core-sets for the MEB problem:

\begin{theo}[Size of $\varepsilon$-core-sets for MEB] \label{theo:coreSetsMEB}%
 Let $P \subseteq \R^d$ be compact and $\varepsilon > 0$. If $C=\B^d$, then there exists an $\varepsilon$-core-set of $P$ of size at most 
 $$\left \lceil \frac{1}{2\varepsilon+\varepsilon^2}\right\rceil+1,$$  
 and this is the best possible $d$-independent bound.
\end{theo}

In the following section, we will explain our notation, state the basic definitions and collect the tools that we need in order to prove Theorems \ref{theo:main} and \ref{theo:coreSetsMEB}. Section 3 then  proves the mentioned radius identities that will lead to Theorem \ref{theo:coreSetsMEB}. Finally, section 4 is dedicated to the derivation of Theorem \ref{theo:main}.

\section{Geometric Foundations} \label{sec:coreRadiiBasics}%
\subsection{Notation}
Throughout this paper, we are working in $d$-dimensional real space and 
for $A \subseteq
\R^d$ we write $\aff(A), \conv(A), \int(A)$, and $\bd(A)$ for the affine hull, the convex hull, the interior, and the boundary of $A$, respectively. For a set $A \subseteq \R^d$, its dimension is $\dim(A):= \dim (\aff(A))$. Furthermore, for any two sets $A, B \subset \R^d$ and $\rho \in \R$, let $\rho A := \{\rho a: a \in A\}$ and $A+B:= \{a+b: a\in A, b\in B\}$ the $\rho$-\emph{dilatation} of $A$ and the \emph{Minkowski sum} of $A$ and $B$, respectively. 
For short, we abbreviate  $A + (-B)$ by $A -B$ and  $A+\{c\}$ by $A+c$. At some points, we will make use of the identity $A+A= 2A$ for $A\subseteq \R^d$ convex.

Furthermore, $\mathcal{L}^d_k$ and $\mathcal{A}^d_k$ denote the family of all $k$-dimensional linear and affine subspaces of $\R^d$, respectively, and $A|F$ is used for the orthogonal projection of $A$ onto $F$ for $F \in \mathcal{A}^d_k$. 

We call $C \subset \R^d$ a \emph{body}, if $C$ is convex and compact, and \emph{container} if it is a body with $0 \in \int(C)$. By $\CC^d, \CC^d_0$ we denote the families of all bodies and all containers, respectively\footnote{Usually, we consider Problem \ref{prb:contHomo} as being parametrized by the container and having varying sets $P$ as input. For feasibility of Problem \ref{prb:contHomo} for a fixed input $P$, it would suffice to impose the condition that $P$ be contained in some affine subspace parallel to $\aff(C)$. As this condition is rather technical and yields no further insight, we restrict to full-dimensional containers; and, as the problem is invariant under translation of the container, we simply assume that $0 \in \int(C)$ for convenience.}.

We write $\B^d:= \{x \in \R^d: \|x\|_2 \leq 1\}$ for the Euclidean unit ball and $x^Ty$ for the standard scalar product of $x,y\in \R^d$. 
By $H^{\leq}_{(a,\beta)}:= \{x \in \R^d: a^T x \leq \beta\}$ we denote the half-space induced by $a\in \R^d$ and $\beta \in \R$, bounded by the hyperplane  $H^{=}_{(a,\beta)}:= \{x \in \R^d: a^T x = \beta\}$. For a convex body $C \in \CC^d$, we write $C^\circ:=\{a \in \R^d: a^T x \leq 1 ~\forall x \in C \}$ for its polar.

For $k \in \{1, \dots, d\}$, a $k$-simplex is the convex hull of $k+1$ affinely independent points. Additionally, let $T^d \in \CC^d$ denote some regular $d$-simplex. Orientation and edge length are not specified as they will be of no interest here.

Finally, for fixed $C\in \CC_0^d$ we denote by $c_P$ a possible center for $P$, i.e. a point such that $P \subseteq c_P +  R(P,C)C$. (Notice, that for general $C$, the center $c_P$ might not be unique.) 

\subsection{Core-Sets and Core-Radii}
As already pointed out in the introduction the concept of $\varepsilon$-core-sets has proved very useful for the special case of the Euclidean MEB problem. Here, we introduce two slightly different definitions for the general {\MCP}: core-sets and center-conform core-sets together with a series of radii closely connected to them. The explicit distinction between the two types of core-sets should help to overcome possible confusion founded in the use of the term core-set for both variants in earlier publications.

\begin{defi}[Core-radii and $\varepsilon$-Core-sets]  \label{defi:coreSetsRadii}%
  For $P \subset \R^d$,  $C \in \CC^d_0$, and $k= 1, \dots, d$, we call $$R_k(P,C):= \max \{R(S,C): S\subseteq P,~ |S| \leq k+1\}$$ 
  the \emph{$k$-th core-radius} of $P$.
  
  \noindent Let $\varepsilon \ge 0$. A subset $S \subseteq P$ such that 
\begin{equation}R(S,C) \le R(P,C) \le (1 + \varepsilon)R(S,C)\end{equation}
  will be called an \emph{$\varepsilon$-core-set} of $P$ (with respect to $C$). 
  
 An $\varepsilon$-core-set $S\subseteq P$ which has the additional property,
 that there exists a center $c_S$ of $S$, such that
\begin{equation}P \subset c_S + (1+\varepsilon)R(S,C)C, \end{equation}
will be called a \emph{center-conform $\varepsilon$-core-set} of $P$ (with respect to $C$).

All three notions are illustrated in Figure \ref{fig:coreSets}.

\end{defi}

\begin{figure}[h]
\centering
\includegraphics[width= 0.3\textwidth]{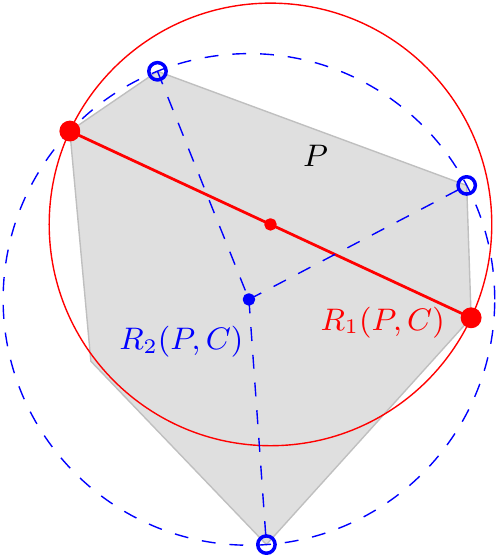}
\hfill
\includegraphics[width= 0.3\textwidth]{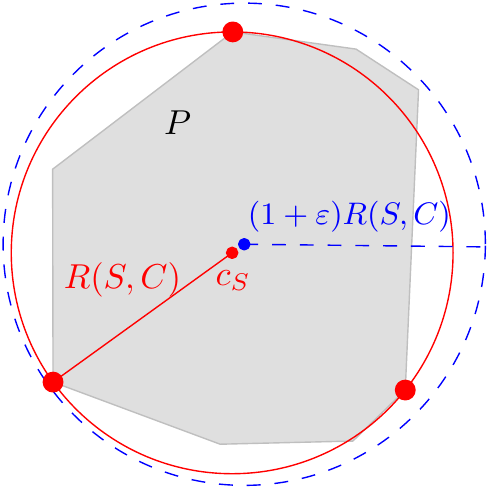}
\hfill
\includegraphics[width= 0.3\textwidth]{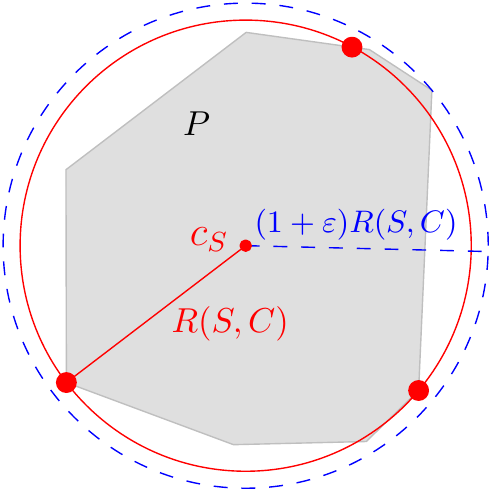}
\caption{Illustration of Core-Radii and Core-Sets. Left: The red filled and the blue empty points define $R_1(P,C)$ and $R_2(P,C)$, respectively. Middle (Right): For some $\varepsilon > 0$, the red filled point set $S$ forms a (center-conform) $\varepsilon$-core-set of $P$ of size 3.  (In all three cases $C= \B^2$.) \label{fig:coreSets}}
\end{figure}

By definition, every center-conform $\varepsilon$-core-set is also an $\varepsilon$-core-set. It will be shown in Lemma \ref{lem:centerConformity} that an $\varepsilon$-core-set is a center-conform $\varepsilon'$-core-set for an $\varepsilon'$ slightly greater than $\varepsilon$, if $C=\B^d$.
 
 Surely, if one is only interested in an approximation of $R(P,C)$ the knowledge of a good core-set
 suffices. A center-conform core-set $S$ carries the additional information of a
 center $c_{S}$ of $S$ that can be used to cover $P$. However, if the center
 of $S$ is not unique, it may not be possible to actually determine which of the
 centers of $S$ is suitable, when $S$ is the only information about $P$ to be
 considered (cf. Remark \ref{rem:centerConformBoxes}).

 We present lower bounds on the sizes of core-sets (these are also lower bounds on the size of center-conform core-sets),
 and we note that most existing positive results (via construction algorithms) already hold for center-conform core-sets. 
 When searching for lower bounds, we use the fact that there exist $\varepsilon$-core-sets of size at most $k+1$ 
 if and only if the ratio ${R(P,C)}/{R_k(P,C)}$ is less than or equal to $1+\varepsilon$. This allows us to transfer the size-of-core-sets problem to bounding the ratio between the core-radii of $P$.

\medskip As already observed in \cite{gritzmannKlee-92}, the reason for restricting the core-radii to $k \le d$ follows directly from Helly's Theorem 
 (see e.g. \cite{dgk-63}). We need a slightly more general statement here, which we prove in the following lemma for completeness. However, the main part of the proof is parallel to the one in \cite{gritzmannKlee-92} for balls. Figure \ref{fig:0CoreSets} also illustrates the situation.

\begin{lem}[0-Core-Sets] \label{lem:0CoreSets}%
Let $P\in \CC^d$, $C\in \CC_0^d$ and $\dim(P) \le k \le d$. Then  $R_k(P,C) =
R(P,C)$, i.e. there exist (center-conform) 0-core-sets of size at most $\dim(P)+1$ for all $P$ and $C$.

Furthermore, for $k\leq \dim(P)$, there always exists a simplex $S\subseteq P$ such that $\dim(S)= k$ and $R(S,C)= R_k(P,C)$.
\end{lem}

\begin{proof}
  Clearly, $R_k(P,C) \le R(P,C)$. To show $R_k(P,C)\ge R(P,C)$ for $k \ge \dim(P)$, observe that by definition of $R_k(P,C)$, every $S \subset P$ 
  with $|S| \le k+1$ can be covered by a copy of $R_k(P,C)C$. This means $ \bigcap_{p\in S} (p - R_k(P,C) C )\neq \emptyset $ for all such $S$.
  Now, as the sets $p - R_k(P,C) C$ are compact, Helly's Theorem applied within $\aff(P)$ yields $\bigcap_{p\in P} (p - R_k(P,C)C) \neq \emptyset$.
  Thus the whole set $P$ can be covered by a single copy of $R_k(P,C) C$.  
Moreover, by applying Helly's Theorem within $\aff(S)$ one may always assume that the finite set $S$ with $R(S,C)= R_k(P,C)$ is affinely independent. Hence,
if $|S|\leq k \leq \dim(P)$ one may complete $S$ to the vertex set of a $k$-dimensional simplex within $P$.
\end{proof}
\begin{figure}[h]
\centering
\includegraphics[width=0.49\textwidth]{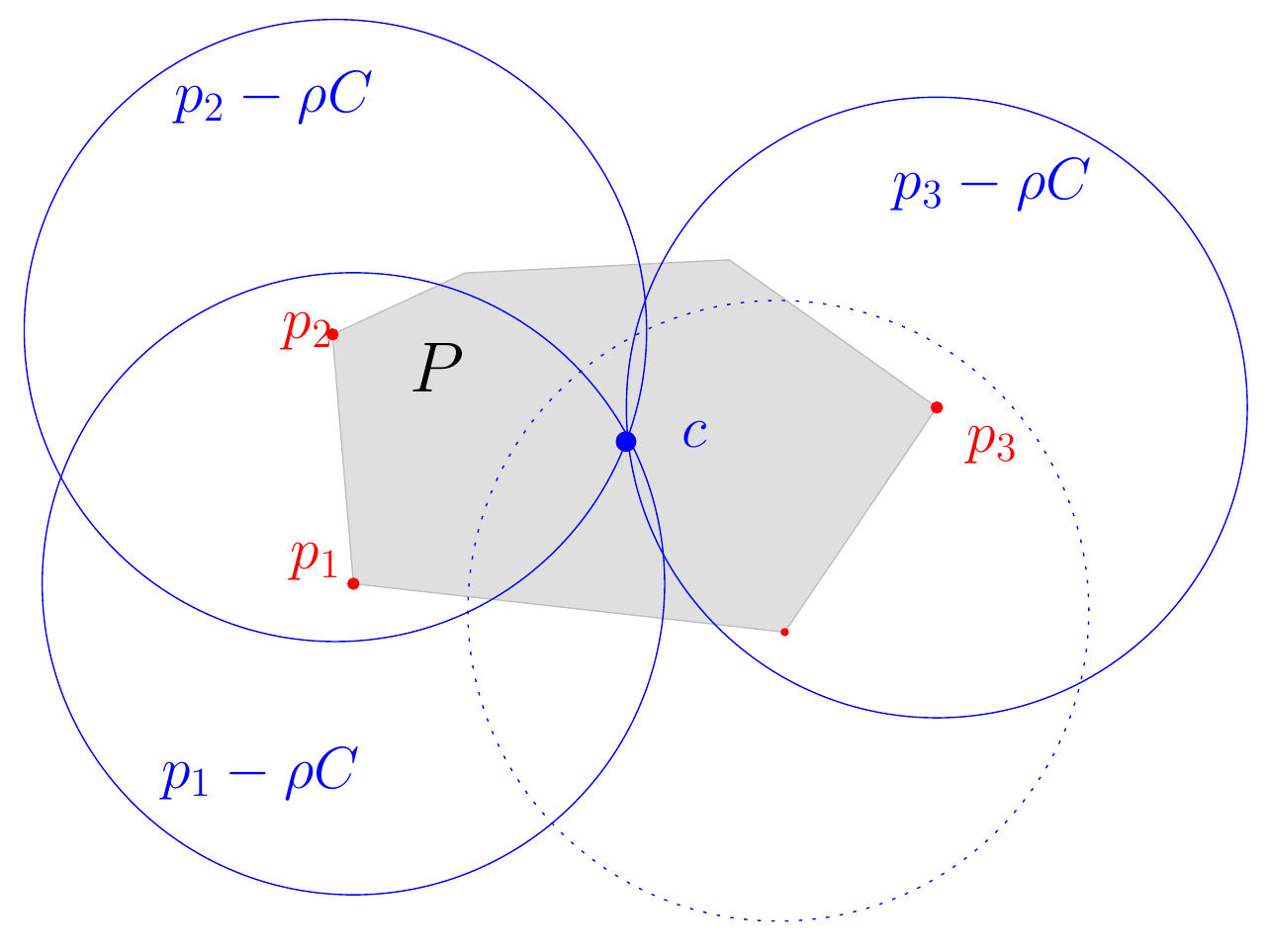}
\hfill
\includegraphics[width=0.49\textwidth]{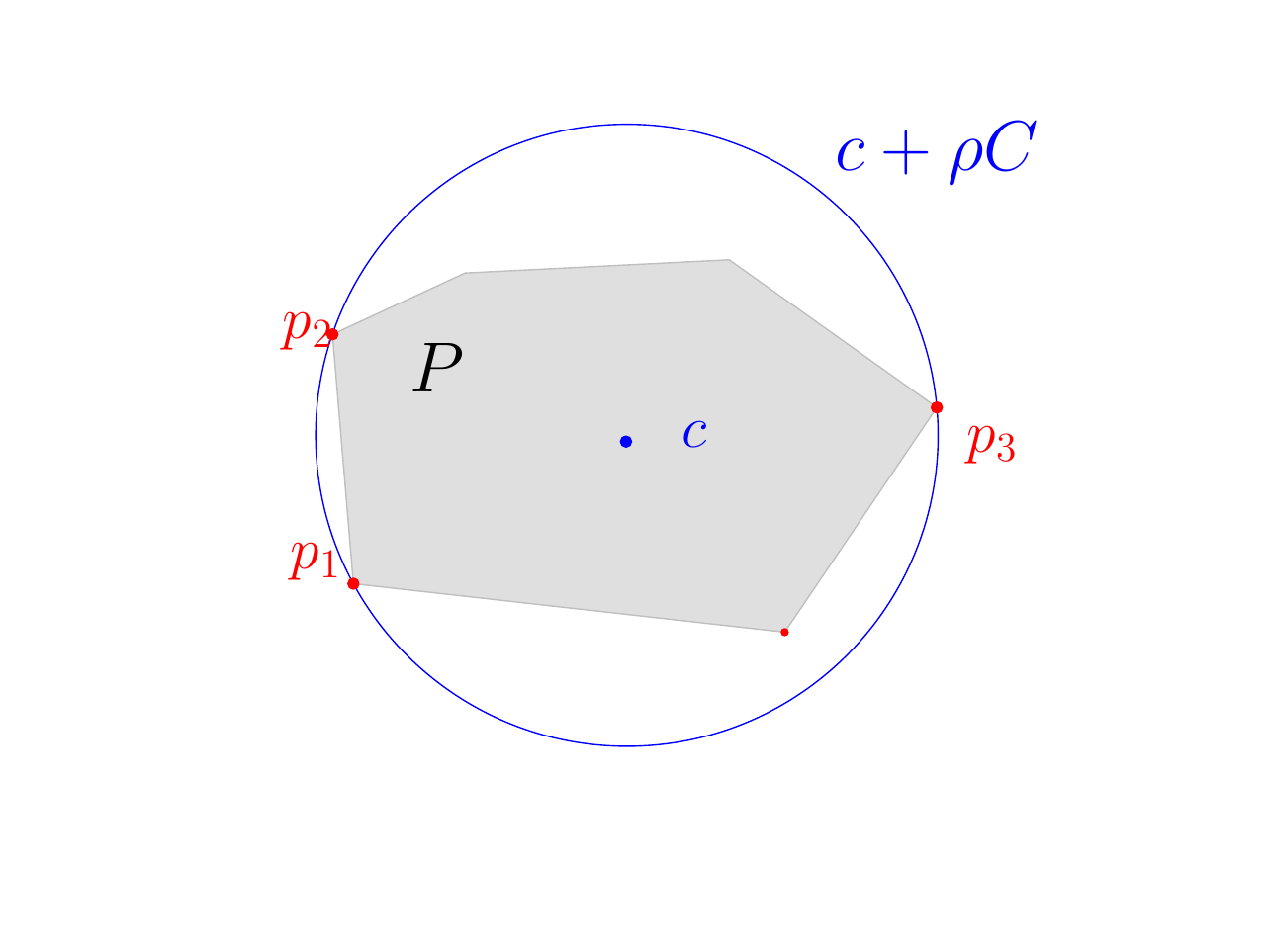}
\caption{The duality argument that makes Helly's Theorem applicable for Containment under Homothetics: $ \bigcap_{p\in S} (p -\rho C )\neq \emptyset \Leftrightarrow R(S,C) \leq \rho$. \label{fig:0CoreSets}}
\end{figure}
\medskip

\subsection{Optimality Conditions}
A characterization of optimal solutions for the MEB case of the {\MCP} can already be 
found in \cite{bonnesenFenchelT}. A corollary, known as 'half-space lemma', 
proved very useful in the construction of fast algorithms for MEB
(see, e.g. \cite{badoiu,bhpi-02,halbraumlemma}). 
However, to our knowledge, the literature does not contain any explicit optimality 
conditions for the general \MCP. 

For brevity, $P$ is said to be \emph{optimally contained} in $C$, if $P \subseteq C$ but there is no $c \in \R^d$ 
and $\rho < 1$ such that $P \subseteq c + \rho C$.  

\begin{theo}[Optimality condition for Problem \ref{prb:contHomo}] \label{theo:optCond}%
Let $P\in \CC^d$ and $C \in \CC_0^d$. Then $P$ is optimally contained in $C$ if and only if 
\begin{enumerate}[(i)]
\item $P \subseteq C$ and
\item for some $2 \le k \le d+1$, there exist $p_1, \dots, p_k \in P$ and hyperplanes 
 $H^=_{(a_i, 1)}$ supporting $P$ and $C$ in $p_i$, $i=1,\dots,k$ such that $0 \in \conv\{a_1, \dots, a_k\}$.
\end{enumerate}
The theorem stays valid even if one allows $C$ to be unbounded. 
\end{theo}

\begin{figure}[h]
\centering
\includegraphics[width=0.4\textwidth]{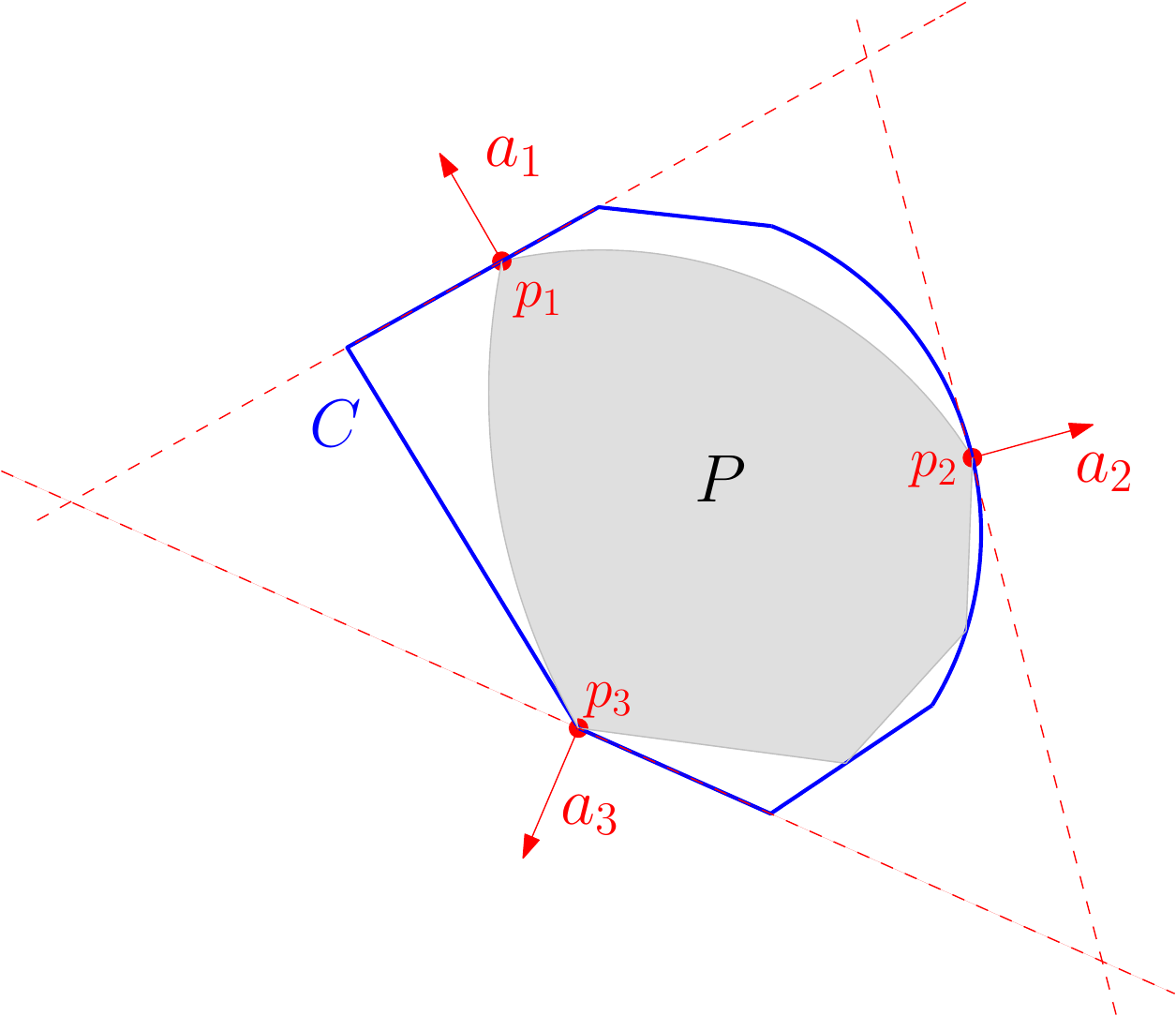}

\caption{The necessary and sufficient conditions from Theorem \ref{theo:optCond}: Condition (ii) is fulfilled by the three points $p_1, p_2, p_3$ and the three hyperplanes with outer normals $a_1,a_2,a_3$, all highlighted in dashed red. Note that (in general) $\conv\{p_1,\dots,p_k\}$ is optimally contained in $\bigcap_{i=1}^k H^\le_{(a_i, 1)}$.} \label{fig:optCond}
\end{figure}

\begin{proof}
Let $C \in \CC_0^d$ be given as $C= \bigcap_{a\in N} H^{\leq}_{(a,1)}$ where $N=\bd(C^\circ)$ is the set of outer normals of $C$. 

First, assume (i) and (ii) hold. By (i), $R(P,C)\leq 1$. Now suppose $R(P,C) <1$. Then there exists $c\in \R^d$ 
and $0< \rho < 1$ such that $c + P\subseteq \rho C$. From (ii) follows $P \cap \bd(C) \neq \emptyset$ and therefore 
$c \neq 0$. Moreover, as $c + P\subseteq \rho C$, it follows $\sup_{a\in N} a^T (c+p_i)\leq \rho$ and in 
particular, $a_i^T (c+p_i)\leq \rho < 1$ for all $i$. Now, as $0\in\conv\{a_1, \dots, a_k\}$, there exist 
$\lambda_i \geq 0$ with $\sum_i \lambda_i =1$ such that $\sum_i \lambda_i a_i=0$ and $\sum_i \lambda_i a_i^T (c+ p_i) < 1$.
Using $a_i^Tp_i=1$ one obtains $\sum_i \lambda_i a_i^T c < 0$, an obvious contradiction. 
Thus, conditions (i) and (ii) imply optimality.

Now, let $P$ be optimally contained in $C$. The following part of the proof is illustrated in Figure \ref{fig:proofOptCond}. As $C$ is compact, we can apply Lemma \ref{lem:0CoreSets} which yields $k\leq d+1$ points $p_i \in P\cap \bd(C)$ for $i=1,\dots, k$ such that 
\begin{equation}
  R(\conv\{p_1,\dots, p_k\}, C) = 1.
  \label{eq:r=1}
\end{equation} 
Let $A = \{ a\in N: \exists i \in \{1, \dots, k\} ~s.t.~ a^T p_i =1 \}$. Since $P \subset C$, for $a \in A$, we have that $a^T p\leq 1$ for all $p\in P$, and $a^T p_i = 1$ for at least one $i$ by definition of $A$. We will show that $0 \in \conv(A)$. 
The statement that there exists a set of at most $d+1$ outer normals with $0$ in their convex hull then follows from 
Caratheodory's Theorem (see \cite{dgk-63}). Assume, for a contradiction, 
that $0\not \in \conv(A)$. Then $0$ can be strictly separated from $\conv(A)$, i.e. there exists $y\in\R^d$ with 
$a^T y \geq 1$ for all $a\in A$. Now, for $A' = \{a\in N: a^T y \leq 0\}$ there exists $\varepsilon > 0$ such that 
$(A'+ \varepsilon\B^d) \cap A = \emptyset$, i.e. $a^T p_i < 1 - \varepsilon$ for all $a\in A'$ and therefore 
$$a^T (p_i - \frac{\varepsilon}{\|a\|\|y\|} y) = a^T p_i + \varepsilon \frac{-a^T y}{\|a\|\|y\|} < 1.$$ 
Moreover, if $a\in N \setminus A'$ then $$a^T (p_i - \frac{\varepsilon}{\|a\|\|y\|} y) = a^T p_i -\varepsilon 
\frac{a^T y}{\|a\|\|y\|} < 1.$$ 

As $0\in \int(C)$, we know that $N\subseteq C^\circ$ is bounded and therefore there exists $\alpha > 0$ such that $\|a\| \leq \alpha$ for all $a\in N$.
Thus, altogether, $p_i - \frac{\varepsilon}{\alpha\|y\|} y \in \int(C)$ 
for all $i$, which contradicts (\ref{eq:r=1}).
\begin{figure}[h]
\centering
\includegraphics[width=0.9\textwidth]{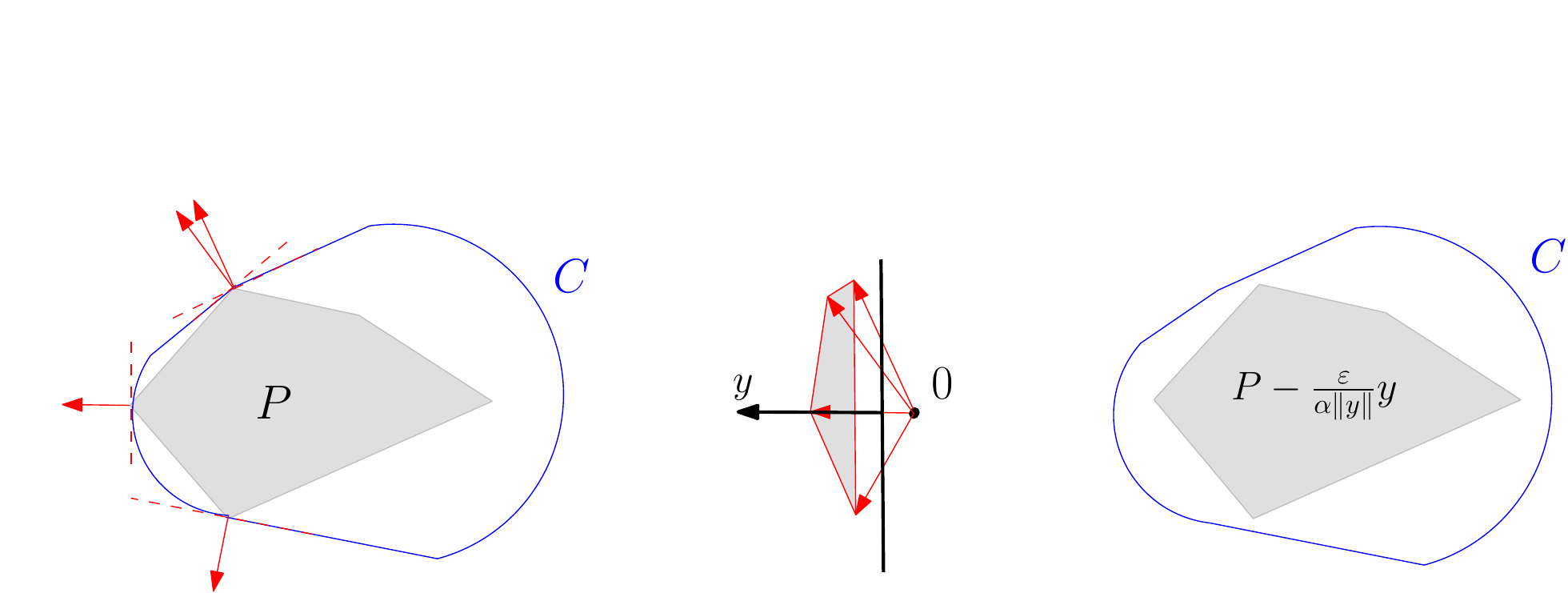}

\caption{The idea of the proof of Theorem \ref{theo:optCond}: If the outer normals in the points where $P$ touches $C$ do not contain the origin in their convex hull, the separation theorem yield a direction $y \in \R^d \backslash\{0\}$ such that $P - \lambda y \subseteq \int(C)$ for a sufficiently small $\lambda \geq 0$.\label{fig:proofOptCond}}
\end{figure}
Finally, observe that the last statement about a possibly unbounded $C$ can be obtained from the one for bounded containers by 
considering a new container $C'= C \cap C''$ where $C'' \in \CC^d_0$ such that $P \subset C''$ and
$P \cap \bd(C'') = \emptyset$.
\end{proof}

\medskip

{\bf Remark:} Besides the direct geometric proof of Theorem \ref{theo:optCond} as stated above, 
 it is also possible to derive the result from the Karush-Kuhn-Tucker conditions in convex optimization (see e.g. \cite[Corollary 28.3.1]{rockafellar}).

As we assume that $C$ has non-empty interior, '$P$ optimally contained in $C$' implicitly implies $|P| > 1$. So, in case $P= \{p\}$, Theorem \ref{theo:optCond} is not applicable and we note for completeness, that in this case, $P$ is optimally contained in $p + 0\cdot C$.

\begin{cor} \label{cor:optCond}%
 Let $P \in \CC^d$ and $C$ a polytope in $\R^d$. If $P\subseteq C$ and $P$ touches every facet of $C$, 
 then $P$ is optimally contained in $C$.
\end{cor}
\begin{proof}
If $C$ is a polytope with facets $F_i = C \cap H^{=}_{(a_i, 1)}$, $i=1,\dots, m$, it is well known 
\cite{bonnesenFenchelT} that, with the choice $\lambda_i = \vol_{d-1}(F_i)$, one has 
$\sum_{i=1}^m \lambda_i a_i = 0$.
\end{proof}

\begin{cor}[Optimality condition for the MEB problem / Half-space lemma] \label{cor:optCondEuclidean}%
Let $P \in \CC^d$. If $P\subseteq \B^d$, then the following are equivalent:

\begin{enumerate}[(i)]
\item $R(P, \B^d)= 1$.
\item For some $k\leq d+1$, there exist $p_1, \dots, p_k \in P \cap \bd(\B^d)$ such that $0 \in \conv \{p_1,\dots, p_k\}$.
\item $0$ can not be strictly separated from $P \cap \bd(\B^d)$.
\item $P \cap \bd(\B^d) \cap H \neq \emptyset$ for every half-space $H$ containing the origin in its boundary. (\emph{Half-space lemma})
\end{enumerate}
\end{cor}

\subsection{Side Notes}
\begin{lem}[Center-conformity for MEB] \label{lem:centerConformity}%
If $P\in \CC^d$, $\varepsilon > 0$ and $S\subseteq P$ is an $\varepsilon$-core-set of $P$ with respect to $\B^d$,  then $S$ is also a center-conform 
     $(\varepsilon+\sqrt{2\varepsilon+\varepsilon^2})$-core-set of $P$.
\end{lem}

\begin{proof}
Let $p \in P$ such that $\max_{x \in P}\|c_S-x\|_2 = \|c_S - p\|_2$. Further let $H$ be a hyperplane perpendicular to $\aff\{c_S, c_P\}$ passing through $c_S$.
Denote by $H^-$ the halfspace which is bounded by $H$ and does not contain $c_P$. Then by Corollary \ref{cor:optCondEuclidean}, there is a point $q \in S \cap H^-$ at distance $R(S,\B^d)$ of $c_S$.
Hence  
\begin{equation*}
\| c_P- c_S \|_2^2   \le \| c_P- q\|_2^2 - \|q - c_S\|_2^2  
  \le R(P, \B^d)^2 - R(S, \B^d)^2 \le (2\varepsilon+\varepsilon^2) R(S, \B^d)^2 
\end{equation*}
and
\begin{equation*}
     \|c_S-p\| \le \|c_S - c_P\| + \|c_p-p\| \le \sqrt{2\varepsilon+\varepsilon^2}
   R(S,\B^d) + R(P, \B^d) 
    = (1+\varepsilon+\sqrt{2\varepsilon+\varepsilon^2}) R(S, \B^d).
\end{equation*}
\end{proof} 

\bigskip

Choosing $P=-C$, one sees that the \emph{Minkowski asymmetry} (or the reciprocal of Minkowski's measure of symmetry, \cite[Note 14 for Section 3.1]{schneider})
\begin{equation} \label{eq:minkowskiAsymm}
 s(C):= R(-C,C)
\end{equation}
can also be expressed as a special case of containment under homothetics. 

\medskip As a further corollary of Theorem \ref{theo:optCond}, we present a very transparent proof (due to \cite{swanepoel-brandenberg-private-06}) 
for the well known fact that the Minkowski asymmetry of a body $C$ is bounded from above by $\dim(C)$. We will make use of the sharpness condition of Corollary \ref{cor:asymm} to show the sharpness of the inequality in Theorem \ref{theo:asymmIneq}.

\begin{cor}[Maximal asymmetry] \label{cor:asymm}%
For every $C\in \CC^d$, the inequalities $1 \le s(C) \le \dim(C)$ hold, with equality, if $C$ is 0-symmetric in the first and if $C$ is a $d$-simplex in the latter case.
\end{cor}
\begin{proof}
Clearly, the Minkowski asymmetry is bounded from below by 1 and $s(C)=1$ if $C=-C$. For the upper bound we suppose (without loss of generality) that $C$ is full-dimensional. 
 Then Lemma \ref{lem:0CoreSets} yields a $d$-simplex $S \subseteq C$ such that $s(C)= R(-C,C)=  R(-S, C) \le R(-S, S) = s(S)$. Thus, it suffices to show $s(S) = \dim(S)$ for every simplex $S$.
 Suppose $S= \conv\{x_1,\dots, x_{d+1}\} \subseteq \R^d$ is a $d$-simplex, without loss of generality such that $\sum_{i=1}^{d+1} x_i = 0$. 
 For all $i=1, \dots, d+1$, the center of the facet 
 $F_j = d \cdot \conv \{x_i, i \neq j\}$ of $dS$ is $c_j=  \sum_{i \neq j} x_i = - x_j$. Hence $-S \subseteq dS$ and $-S$ touches every facet of $dS$, showing the optimality of the containment by Corollary \ref{cor:optCond}.
\end{proof}

\medskip
{\bf Remark:} In \cite{gruenbaum-63} also the `only if' direction for the sharpness of the bounds in Corollary \ref{cor:asymm} is shown.

\bigskip

Finally, note that Lemma \ref{lem:0CoreSets} can also be seen as a result bounding the combinatorial dimension of Problem \ref{prb:contHomo} interpreted as a Generalized Linear Program (GLP). As it is not our main focus here, we simply mention the connection and refer to \cite{GLP2, GLP} for a treatise on GLPs and to \cite{amenta-thesis} for their relation to Helly-type theorems.

\section{Radii Identities and Small Core-Sets}

\subsection{Identities between Different Radii}

In this section, we show the identity of the core-radii from Definition \ref{defi:coreSetsRadii} to
two series of intersection- and cylinder/projection-radii in convex geometry,
similar to the ones defined in \cite{henk92} and to the more often considered
ones in \cite{gritzmannKlee-92} and \cite{pukhov-80}.  This identity will help us to use a set of known geometric inequalities on these radii to obtain bounds on core-set sizes.

\begin{defi}[Intersection- and Cylinder-Radii] \label{def:henk-radii}%
 For $P\in\CC^d$, $C\in \CC_0^d$ and $k \in \{1,\dots, d\}$, let
 \begin{align} \label{eq:cut}
   R^{\sigma}_k(P,C):= \max \{ R(P \cap E, C) : E \in \mathcal{A}^d_k\} \\
   \intertext{and}  
   \label{eq:cylinder}
   R^{\pi}_k(P,C):= \max \{ R(P , C+F) : F \in \mathcal{L}^d_{d-k} \}
 \end{align}
 Notice, that, as $C+F$ is unbounded, $R(P, C+F)$ is a slight abuse of notation. It follows from Blaschke's Selection Theorem \cite{schneider}, that the maxima in \eqref{eq:cut} and \eqref{eq:cylinder} exist.
\end{defi}

\begin{rem}[Cylinder-radii in Euclidean spaces]
When dealing with the Euclidean unit ball $\B^d$, the observation that, for $F \in \mathcal{L}^d_{d-k}$, $R(P, \B^d+F)=R(P|F^\perp, \B^d)$ shows that the cylinder-radii can be interpreted as projection-radii, i.e.
$$ R^{\pi}_k(P,\B^d)= \max \{R(P|E, \B^d): E\in \CL_k^d\} \ .$$
\end{rem}

\bigskip 

The following theorem states the identity of these three series of radii. To the best of our knowledge, even the equality between the intersection- and projection-radii in the Euclidean case has not been shown before. 

\begin{theo}[Identity of Intersection-, Cylinder- and Core-Radii] \label{theo:differentRadii}%
Let $P\in\CC^d$, $C\in \CC_0^d$ and $k \in \{1, \dots, d\}$. Then, 
  $$ R_k(P, C) = R^{\sigma}_k(P,C) = R^{\pi}_k(P,C). $$
\end{theo}

\begin{proof}
We show $R_k(P,C) \le R_k^\sigma(P,C)\le R_k^\pi(P,C)\leq R_k(P,C)$.

 First, $R_k(P,C) \le R_k^\sigma(P,C)$: By definition of the core-radii, there exists $S \subseteq P$ with $|S|= k+1$ and 
 $R(S,C)=R_k(P, C)$. Since $\dim(\aff(S)) \le k$, one obtains $$R_k(P, C)= R(S, C) \leq R( P \cap \aff(S), C)\leq R_k^{\sigma}(P,C).$$

Now, $R_k^\sigma(P,C)\le R_k^\pi(P,C) $: Let $E \in \CL_{k}^d$ such that $R_k^\sigma(P,C)= R(P\cap E,C)$. As $\dim(P\cap E)\leq k$, Lemma \ref{lem:0CoreSets} and Theorem \ref{theo:optCond} show that, for $m \leq k+1$, there are points $p_1,\dots, p_{m}\in P \cap E$ and hyperplanes $H^=_{(a_1, 1)}, \dots, H^=_{(a_m, 1)}$ such that $H^=_{(a_i, 1)}$ supports $C$ in $p_i$ and $0 \in \conv \{a_1,\dots, a_m\}$. As $0 \in \conv\{a_1,\dots, a_m\}$, we get that $\dim \{a_1,\dots, a_m\}^\perp\geq d-k$ and we may choose $F \in \CL_{d-k}^k$ such that $F \subseteq \{a_1,\dots, a_m\}^\perp$. Again by Theorem \ref{theo:optCond}, if follows that
$$ R_k^\sigma(P,C)= R(P\cap E,C) = R(P \cap E, C + F) \leq R(P, C+F) \leq R_k^\pi(P,C). $$

 Finally, $R^{\pi}_k(P,C) \leq R_k(P, C)$: Let $F \in \CL_{d-k}^d$ such that $R_k^\pi(P,C)= R(P,C+F)$ 
 and suppose without loss of generality that $P$ is optimally contained in $C+F$ 
 (i.e. the optimal radius and center are $\rho^*=1$ and $c^*=0$, respectively). 
 Then it follows from the statement for unbounded containers in Theorem \ref{theo:optCond} that there exist $m \le d+1$ points $p_1,\dots, p_m \in P$ 
 and hyperplanes $H^=_{(a_i, 1)}$, $i=1, \dots, m$ 
 such that $H^=_{(a_i, 1)}$ supports $C+F$ in $p_i$ and $0 \in \conv \{a_1,\dots, a_m\}$. 
 Since every direction in $F$ is an unbounded direction in $C+F$, one obtains $a_i \in F^\perp$ for all $i=1,\dots, m$. 
 Now, by Caratheodory's Theorem, there exists a subset $I \subseteq \{1,\dots , m\}$ with $|I| \le \dim(F^\perp) + 1 = k+1$ 
 such that $0\in \conv\{a_i: i\in I\}$. Applying again Theorem \ref{theo:optCond}, $$R_k^\pi (P,C) = R(P,C+F) 
 = R(\conv\{p_i: i\in I\}, C+F) \leq R(\conv\{p_i : i\in I\}, C) \leq R_k(P,C).$$
\end{proof}

\subsection{Dimension Independence for Two Special Container Classes} \label{subs:positive}
 The most evident (non-trivial) example for a restricted class of containers allowing small core-sets may be parallelotopes. 
 E.g. in \cite[\S{25}]{excursionsIntoCombGeom}, the following Proposition is shown:

 \begin{prop}[Core-radii for parallelotopes]
   The identity 
   $$R_1(P, C) = R(P,C)$$
   holds true for all $P \in \CC^d$ if and only if $C\in \CC_0^d$ is a parallelotope. 
 \end{prop}

 In terms of core-sets, this means that there is a 0-core-set of size two for all $P \in \CC^d$, if $C$ is a parallelotope 
 and that these are the only containers with this property.

\begin{rem}\label{rem:centerConformBoxes}%
 Even though there always exist center-conform 0-core-sets of size two, if $C$ is a
 parallelotope, one may need $S$ to contain at least $d+1$ points in order to actually
 identify which of the centers $c_{S}$ of $S$ may be a center of $P$.
 Let, e.g., $\tau \in [-1,1]$, $P=\conv \{(\tau \pm 1)u_{1},\dots,(\tau \pm
 1)u_{d-1},\pm u_{d}\}$, and $C=[-1,1]^{d}$. Then $S= \{-u_{d}, u_{d}\}$ is a 0-core-set
 of $P$ with respect to $C$ and every point in $[-1,1]^{d-1} \times \{0\}$ a
 possible center $c_{S}$ of $S$ (indicated by the red arrow in Figure \ref{fig:boxCenters}).
 As soon as one decides for one such $c_{S}$ to form a center-conform $\varepsilon$-core-set
 $S$ without further knowledge of $P$ (in our case of $\tau$), there may be a point $p_{3}$ in $P
 \setminus (c_{S} + (1+\varepsilon)R(S,C)C)$ for any $\varepsilon < 1$. 

\begin{figure}[h]
\centering
\includegraphics[width=0.45\textwidth]{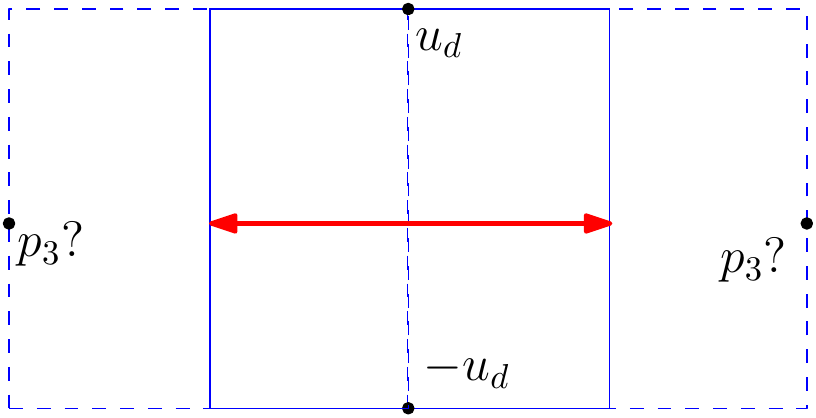}
\caption{The situation of the example of Remark \ref{rem:centerConformBoxes}. \label{fig:boxCenters} }
\end{figure}
\end{rem}
\medskip

 Surely, a more important restricted class of containers is the class of ellipsoids. 
 In \cite{henk92} geometric inequalities are derived which relate the radii of Definition \ref{def:henk-radii} within each series. Using Theorem \ref{theo:differentRadii}, these inequalities can be presented in a unified way in terms of core-radii:

 \begin{prop}[Henk's Inequality] \label{prop:henk}%
   Let $P\in \CC^d$ and  $k,l \in \N$ where $l\leq k \leq d$. 
   Then
   \begin{equation} \label{eq:henk} \frac{R_k(P, \B^d)}{R_l(P, \B^d)} \leq  \sqrt{\frac{k(l+1)}{l(k+1)}} \end{equation}
   with equality if $P=T^d$.
  
 \end{prop}
\medskip
 {\bf Remark:} Because of the affine invariance of \eqref{eq:henk} one may replace $\B^d$ by any $d$-dimensional ellipsoid.
 
\bigskip

This inequality can now directly be turned into a sharp bound on the size of $\varepsilon$-core-sets for the MEB problem and Theorem \ref{theo:coreSetsMEB} follows:

\begin{proof}[Theorem \ref{theo:coreSetsMEB}]
 Let $\varepsilon > 0, k = \left\lceil \frac{1}{2 \varepsilon+\varepsilon^2 } \right\rceil$, and $S\subseteq P$ such that $R(S,\B^d)=R_k(P,\B^d)$. 
 Then $|S| \leq k+1$ and by Proposition \ref{prop:henk} and Lemma \ref{lem:0CoreSets}:
 $$ {R(P,\B^d)} \le \sqrt{\frac{d(k+1)}{k(d+1)}} \cdot R(S,\B^d)$$ 
 where $k$ is chosen such that $\sqrt{\frac{d(k+1)}{k(d+1)}} \le 1+ \varepsilon$ independently of $d\in \N$.

 Now, we show the sharpness of the bound: Let $d \in \N$ such that $\frac{d}{d+1} > (1+\varepsilon)^2 \frac{k}{k+1}$ and 
 choose $P= T^d$. Now, for $k < \frac{1}{2\varepsilon + \varepsilon^2}$ if $S' \subseteq P$ consists of no more than $k+1$ 
 points then 
$$  R(P, \B^d) = \sqrt{\frac{d(k+1)}{k(d+1)}}R_k(T^d, \B^d) 
   > (1+\varepsilon) R_k(T^d, \B^d) \geq (1+\varepsilon)R(S', \B^d). $$
 Hence $S'$ is not an $\varepsilon$-core-set of $P$.
\end{proof}

\medskip

{\bf Remark:} Jung's well known inequality \cite{jung-01}, relating the diameter and the outer radius of $P$, can be obtained from Proposition \ref{prop:henk} just by choosing $k=d$ and $l=1$. As Proposition \ref{prop:henk}, it can be turned into a core-set result saying that, for the Euclidean ball in every dimension, a diametral pair of points in $P$ is already a $(\sqrt{2} -1)$-core-set.

\medskip

A very easy and intuitive algorithm to actually find $\varepsilon$-core-sets of a finite set $P$ was first introduced in \cite{bhpi-02}. Roughly speaking, 
it starts with a subset $S\subset P$ of two (good) points and computes (or approximates) the minimum enclosing ball $B_S$ for $S$. 
Whenever a dilatation by $(1+\varepsilon)$ of $B_S$ centred at $c_S$ does not cover the whole set $P$, 
an uncovered point is added to $S$ and the process is 
iterated. The analysis in \cite{bhpi-02} shows that this algorithm produces $\varepsilon$-core-sets of size $O(1/\varepsilon^{2})$, and, by construction, these are even center-conform.

In \cite{BadoiuClarksonOptimal}, the existence of center-conform $\varepsilon$-core-sets of size $1/{\varepsilon}$ and the sharpness of this bound are shown.
Theorem \ref{theo:coreSetsMEB} now complements this result and gives a tight upper bound on the size of (general) core-sets, which is roughly half the center-conform bound.

\section{No Sublinear $\varepsilon$-Core-Sets}

 In this section several geometric inequalities between the core-radii are collected and then used to derive positive and negative results on possible $\varepsilon$-core-set sizes. One should remember that, because of Lemma \ref{lem:0CoreSets}, we already know the existence of 0-core-sets 
 of size $d+1$, i.e. 
 not depending on the size of $P$ (nor $C$) and only linearly depending on $d$.

\subsection{General (non-symmetric) Containers}

 \begin{theo}[Inequality relating core-radii] \label{theo:asymmIneq}%
   Let $P\in \CC^d$, $C\in \CC_0^d$ and $k,l \in \N$ such that $l\leq k \leq d$. Then
   $$\frac{R_k(P, C)}{R_l(P, C)} \leq \frac{k}{l}$$
   with equality if $P = -C = T^d$. 
 \end{theo}

 \begin{proof}
   It suffices to show 
   \begin{equation}
     \frac{R_k(P, C)}{R_{k-1}(P, C)} \le \frac{k}{k-1}.
     \label{eq:generalIneq}
   \end{equation}
   as for $l < k-1$ the claim follows by repeatedly applying \eqref{eq:generalIneq}. 
   Without loss of generality one may assume the existence of a $k$-simplex $S = \conv\{x_1, \dots , x_{k+1}\} \subseteq P$ 
   satisfying $R(S,C)= R_k(P,C)$, as (\ref{eq:generalIneq}) is certainly fulfilled if  ${R_k(P, C)}={R_{k-1}(P, C)}$. 
   Moreover, it can also be supposed that $\sum_{i=1}^{k+1} x_i = 0$ and $R_{k-1}(S, C) = 1$. 
   Now, let $S_j = \conv \{x_i : i\neq j\}$, $j= 1,\dots , k+1$ denote the facets of $S$. 

   Since $\sum_{i=1}^{k+1} x_i =0$, it follows $ -1/k \cdot x_j = 1/k \sum_{i\neq j}  x_i \in \conv\{x_i: i\neq j\} = S_j$ 
   for all $j$ and surely $x_j \in S_i$ for all $i,j$, $i \neq j$. 

Since $R_{k-1}(S, C) = 1$, 
   there exist translation vectors $c_j \in \R^d$ such that $S_j \subseteq c_j + C$ for all $j \in \{1,\dots , k+1\}$ which implies 
   $$ (k-\frac{1}{k}) x_j \in \sum\limits_{i=1}^{k+1} S_i \subset \sum\limits_{i=1}^{k+1} c_i + (k+1) C$$ for all $j$
   and thus $R(S,C) \le (k+1) / (k-\frac{1}{k})$.
   However, since $R_{k-1}(S,C) =1$ we obtain 
   $$R_k(P,C)= R(S,C) \le \frac{k}{k-1} R_{k-1}(S,C) \le \frac{k}{k-1} R_{k-1}(P,C)$$ proving \eqref{eq:generalIneq}.
   
   The sharpness of the inequality for $-P = C = T^d$ follows directly from showing $R_k (T^d, -T^d)= k$ for 
   $k= 1, \dots, d$:
   
   Since every $k$-face $F$ of $T^d$ can be covered by the $k$-face of $-T^d$ parallel to $F$ 
   and since these $k$-faces are regular $k$-simplices, it follows from  Corollary \ref{cor:asymm} 
   that $R(F, -F) = k$ and, thus, $R_k(T^d, -T^d) \le k$ for all $k\in \{1,\dots, d\}$. 

Finally, for every face $F$ of $-T^d$, it is true that $-T^d | \aff(F) = F$. Thus, if $S_k \subseteq T^d$ is a $k$-face of $T^d$ and $S_k \subseteq c+ \rho (-T^d)$ for some $c\in \R^d$ and $\rho \ge 0$, then $S_k | \aff(c + \rho (-S_k)) \subseteq (c+\rho (-S_k))$. However, $\aff(c + \rho (-S_k))$ is parallel to $S_k$, and therefore the above projection is just a translation, which means there exists $c'\in \R^d$ such that 
$S_k \subseteq c'+ \rho (-S_k)$. Using Corollary \ref{cor:asymm} again, 
   it follows that $\rho \ge k$. 
\end{proof}

 \begin{cor}[No sublinear core-sets for general containers] \label{cor:coreSetsAssym}%
   For every $P\in \CC^d, C\in \CC_0^d$ and $\varepsilon \ge 0$, there exists an $\varepsilon$-core-set of $P$ of size at most 
   $\left \lceil \frac{d}{1+\varepsilon}\right\rceil+1$ and for $P = -C = T^d$ no smaller subset of $P$ will suffice.
 \end{cor}
 \begin{proof}
   The case $\varepsilon = 0$ equates to Lemma \ref{lem:0CoreSets}.
   So, let $\varepsilon > 0$ and $k = \left\lceil \frac{d}{1+\varepsilon} \right\rceil$. If 
   $S \subseteq P$ such that $R(S,C)=R_k(P,C)$ then $|S| \le k+1$ and by Theorem \ref{theo:asymmIneq}:
   $$ R(P,C)= R_d(P,C) \le \frac{d}{k} R_k(P,C)= \frac{d}{k}R(S,C).$$ 
   
   By the choice of $k$, ${d}/{k} \le 1+ \varepsilon$.
   
In order to show the sharpness of the bound, choose $P= T^d$ and $C= -T^d$. Now, for $k < \frac{d}{1+\varepsilon}$, if $S' \subseteq P$ consists of no more than $k+1$  points, then it follows from the sharpness condition in Theorem \ref{theo:asymmIneq}, that 

$$  R(T^d, -T^d) = \frac{d}{k}R_k(T^d, -T^d) 
   > (1+\varepsilon) R_k(T^d, -T^d) \geq (1+\varepsilon)R(S', -T^d). $$
Hence $S'$ is no $\varepsilon$-core-set of $P$.
\end{proof}
 
 \medskip
 
 \textbf{Remarks: \\} 
(i) Note that, by Lemma \ref{lem:0CoreSets}, the minimal size of a 0-core-set depends linearly on $d$ and Corollary \ref{cor:coreSetsAssym} now shows that allowing $\varepsilon >0$ does not improve this situation. Thus, Corollary \ref{cor:coreSetsAssym} already proves Theorem \ref{theo:main} for general containers. \\[1ex]

(ii) We would like to mention that, whenever $C$ is a polytope presented as $C = \{x \in \R^d : a_k^Tx \le 1, ~ k=1,\dots , m\}$ and $P=\conv\{p_1,\dots, p_n\}$, Problem \ref{prb:contHomo} can be rewritten as a Linear Program \cite{brandenbergRoth-10, gritzmannKlee-93}, with the help of which a $0$-core-set of $P$ of at most $d+1$ points can be computed in polynomial time.

\begin{rem}[Center-conformity]
Choosing $P=-C=T^d$, every subset $S$ of $d$ vertices of $P$ yields $R(S,C) = d-1$ with a unique center $c_S$. But to cover $P$ by $c_S + \rho C$, we need $\rho \ge \frac{2d}{d-1}R(S,C)$. So, for  $\varepsilon \in (0,1)$, a center-conform $\varepsilon$-core-set may need to be of size $d+1$.
\end{rem}

Moreover, as much as we understood it, \cite[Theorem 5]{panigrahy} asserts (in particular) that for every $\varepsilon > 0$ there is a subset $S \subseteq T^d$ of size $O(1/\varepsilon^2)$ such that every point in $T^d$ has Euclidean distance at most $\varepsilon$ to $c_S + R(S,-T^d)(-T^d)$. Again, taking any subset $S \subseteq T^d$ of $d$ vertices and the fact that the distance of the remaining vertex to $c_S + (d-1)(-T^d)$ is strictly greater than ${1}/{\sqrt{2}}$, shows that this theorem cannot be true for $\varepsilon < {1}/{\sqrt{2}}$.

\subsection{Symmetric Containers/Normed Spaces}
As mentioned in the introduction, every 0-symmetric container $C \in \CC^d_0$ induces a norm $\|\cdot\|_C$ and vice versa. We will always talk about symmetric containers here, but one may as well reformulate all results in terms of Minkowski spaces. 

The results in section \ref{subs:positive} may motivate the hope that symmetry of the container is the key for positive results on dimension-independence.  

Indeed, in \cite{bohnenblust-38}, Bohnenblust proved an equivalent to Jung's Inequality (see the remark after the proof of Theorem \ref{theo:coreSetsMEB}) 
 for general normed spaces. Taking into account the Minkowski asymmetry $s(C)$ of a possibly asymmetric container $C$, 
 a slightly generalized result and a simplified proof are derived in \cite[Lemma 2]{br-07}; in terms of core-radii, it reads as follows:
 
\begin{prop}[Generalized Bohnenblust] \label{prop:bohnenblust}%
Let $P\in \CC^d$, $C\in\CC^d_0$. Then
$$ \frac{R(P, C)}{R_1(P, C)}  \le \frac{(1+s(C))d}{d+1}$$
 with equality, if $P= T^d = -C$ or $P=T^d$ and $C=T^d - T^d$.
\end{prop}

\medskip

One might hope that, for the class of symmetric containers, Bohnenblust's
Inequality can be generalized in the same way as Henk's Inequality generalizes
Jung's in the Euclidean case (giving a bound on the core-radii ratio as in \eqref{eq:false-gen} at the
end of this section). This inequality would be tight for $P=T^d$ and
$C=T^d-T^d$. However, the remainder of this section will show that the bound on the ratio of core-radii with symmetric containers does not improve the general bound from Theorem \ref{theo:asymmIneq}.

 \begin{lem} \label{lem:coreRadii_TdinTdCap-Td}%
   With $C^d = T^d \cap (-T^d)$,
   {\renewcommand{\arraystretch}{1.5}
     \begin{equation}
       R_k(T^d, C^d) = \left\{
         \begin{array}{cl}
           \frac{d+1}{2}  	& \text{ if } k \le \frac{d+1}{2} \\
           k			& \text{ if } k \ge \frac{d+1}{2} \ .
         \end{array} \right.
       \nonumber
     \end{equation} }
 \end{lem}

 \begin{proof}
   Let $T^d= \conv\{x_1,\dots, x_{d+1}\}$ for suitable $x_1, \dots, x_{d+1} \in \R^d$. Independently of which coordinates we choose for $x_1, \dots, x_{d+1}$, we can index the normals $a_1, \dots, a_{d+1} \in \R^d$ of a halfspace presentation of $T^d$ such that $T^d = \bigcap_{i=1}^{d+1} H^{\leq}_{(a_i,1)}$ and 
   $$a_j^T x_i = \left\{
     \begin{array}{cl}
       1	& \text{ if } j\neq i \\
       -d	& \text{ if } j= i
     \end{array}
   \right. \text{ for } i,j\in\{1,\dots, d+1\}
   $$ 

   Let $k\in \{1, \dots, d+1\}$ and consider an arbitrary $k$-face $F$ of $T^d$, without loss of generality, $F= \conv\{x_1, \dots, x_{k+1}\}$.

   For $k \le \frac{d+1}{2}$, let 
   $$\gamma = - \frac{(k-1)(d+1)}{2(d-k)} \text{~~ and ~~} c = \frac{1}{k+1}\left(\sum_{l=1}^{k+1} x_l 
   + \gamma \sum_{l= k+2}^{d+1} x_l\right).$$ 
   Then $\gamma \ge -\frac{d+1}{2}$ and for $i\in \{1,\dots, k+1\}$ and $j\in \{1,\dots, d+1\}$
     {\renewcommand{\arraystretch}{1.5}
       $$a_j^T (x_i-c) = \left\{ 
         \begin{array}{cl}
           \gamma & \text{ if } j > k+1 \\
           \frac{d+1}{2} & \text{ if } j\leq k+1, j \neq i \\
           -\frac{d+1}{2} & \text{ if } j=i \ . \\
         \end{array} \right. $$ }

Hence $F-c \subseteq \frac{d+1}{2}C^d$.
   Moreover, these equalities show that $T^d -c$ touches the facets of $\frac{d+1}{2} C^d$ induced by the hyperplanes 
   $H^=_{(a_i, (d+1)/2)}$, $H^=_{(a_i, -(d+1)/2)}$ for $i=1, \dots, k+1$ and therefore it follows by Theorem \ref{theo:optCond} 
   that $R_k(T^d, C^d)= {(d+1)}/{2}$.
     
   For $k \ge \frac{d+1}{2}$, let $c = \sum_{i=1}^{k+1} x_i$. Then $1-k+d \le k$ and for $i\in \{1,\dots, k+1\}$ and $j\in \{1,\dots, d+1\}$
   $$ a_j^T(x_i-c) = \left \{
     \begin{array}{cl}
       -k 		& \text{ if } j> k+1 \\
       1-k+d		& \text{ if } j\leq k+1, j\neq i \\
       -k		& \text{ if }  j=i 
     \end{array} \right.
   $$
   showing $F-c \subseteq kC^d$.
   Here again, the equalities show that $T^d -c$ touches every facet of $-kT^d$ and $R_k(T^d, C^d)= k$ follows by Theorem \ref{theo:optCond}.
\end{proof}

\begin{theo}[Inequality relating core-radii for 0-sym\-metric containers] \label{theo:ineqSymm}%
  Let  $k,l\in \N$ such that $l\leq k\leq d$, $P\in \CC^d$ and $C\in\CC_0^d$ a 0-symmetric container. Then
  {\renewcommand{\arraystretch}{1.5}
    $$\frac{R_k(P,C)}{R_l(P, C)} \leq \left\{
      \begin{array}{cl}
        \frac{2k}{k+1} & \text{ for } l \le \frac{k+1}{2} \\
        \frac{k}{l}  & \text{ for } l \ge \frac{k+1}{2} \ .
      \end{array} \right.$$ }
  Moreover, let $T^k$ be a $k$-simplex embedded in the first $k$ coordinates of $\R^d$ and $C^k= (T^k \cap (-T^k)) + (\{0\}^k \times [-1,1]^{d-k})$. 
  Then
  {\renewcommand{\arraystretch}{1.5}
    $$\frac{R_k(T^k,C^k)}{R_l(T^k, C^k)} = \left\{
      \begin{array}{cl}
        \frac{2k}{k+1} & \text{ for } l \le \frac{k+1}{2} \\
        \frac{k}{l}  & \text{ for } l \ge \frac{k+1}{2} \ . 
      \end{array} \right.$$ }
\end{theo}

\begin{proof}
Let $S \subseteq P$ be a $k$-simplex such that $R_k(P,C)= R(S,C)$ and assume without loss of generality that $R(S,C)= k$. By Bohnenblust's Inequality, we get that $R_1(S,C) \geq (k+1)/2$ and thus $R_l(P,C)\geq R_1(P,C) \geq (k+1)/2$. Thus $R_k(P,C) / R_l(P,C) \leq 2k/(k+1)$. On the other hand 
$$\frac{R_k(P,C)}{R_l(P,C)} \leq \frac{k}{l}$$
 by Theorem \ref{theo:asymmIneq}. Together this yields $$\frac{R_k(P,C)}{R_l(P,C)} \leq \min \left\{\frac{2k}{k+1}, \frac{k}{l} \right\},$$ which splits into the two cases claimed above.
The second statement follows from Lemma \ref{lem:coreRadii_TdinTdCap-Td} and the observation that the computation of $R(T^k, C^k)$ is in fact the 
$k$-dimensional containment problem of containing $T^k$ in $T^k \cap (-T^k)$.
\end{proof}

\medskip

With Theorem \ref{theo:ineqSymm} at hand, Theorem \ref{theo:main} follows as a simple corollary:

\medskip

\begin{proof}[Theorem \ref{theo:main}]
  For $k=d$ and $l \ge (d+1)/2$ the inequalities in Theorem \ref{theo:asymmIneq} and \ref{theo:ineqSymm} coincide. Hence the proof of 
  Corollary \ref{cor:coreSetsAssym} can simply be copied up to the additional condition that $\varepsilon < 1$ and the change from $C= T^d$ to $C= T^d \cap (-T^d)$ to show that the bound is best possible. 
\end{proof}

\medskip

 On the other hand diametral pairs of points in $P$ are 1-core-sets for every 0-symmetric container $C$ as Bohnenblust's result already shows. 
 Theorem \ref{theo:ineqSymm} then shows that no choice of up to $\lfloor (d-3)/2 \rfloor$ points to add to the core-set improves the approximation quality.

\medskip

{\bf Remark: }  Theorem \ref{theo:main} also implies the 
  non-existence of sublinear center-conform $\varepsilon$-core-sets for $\varepsilon <1$. 
 On the other hand, we know from Lemma \ref{lem:0CoreSets} that there are linear ones, even if $\varepsilon=0$.

\bigskip

Theorem \ref{theo:main} shows that the class of symmetric containers is too large for an extension of Bohnenblust's Inequality to other core-radii than $R_1$. A question that remains open is, whether there is a sensible class of containers $\CC \subseteq \CC_0^d$ such that the following inequality holds for $1\leq l \leq k \leq d$: 
\begin{equation} \label{eq:false-gen} 
 \frac{R_k(P, C)}{R_l(P, C)} \leq \frac{k(l+1)}{l(k+1)} 
\end{equation} 

If true, (\ref{eq:false-gen}) would yield dimension independent $\varepsilon$-core-sets for all $\varepsilon > 0$ in the same way as shown for the MEB problem in the proof of Theorem \ref{theo:coreSetsMEB}. 

\bigskip

{\bf Acknowledgements.}
We would like to thank David Larman for valuable discussions and two anonymous referees for their comments, which helped improving this manuscript.

\bibliographystyle{plain}
\bibliography{references}

\begin{thebibliography}{10}

\bibitem{amenta-thesis}
N.~Amenta.
\newblock Helly-type theorems and generalized linear programming.
\newblock {\em Ph.D. Thesis, University of California at Berkeley}, 1992.

\bibitem{badoiu}
M.~Badoiu and K.~L. Clarkson.
\newblock Smaller core-sets for balls.
\newblock In {\em SODA '03: Proceedings of the 14th annual ACM-SIAM symposium
  on Discrete algorithms}, pages 801--802. Society for Industrial and Applied
  Mathematics, 2003.

\bibitem{BadoiuClarksonOptimal}
M.~Badoiu and K.~L. Clarkson.
\newblock Optimal core-sets for balls.
\newblock {\em Computational Geometry}, 40(1):14 -- 22, 2008.

\bibitem{bhpi-02}
M.~Badoiu, S.~Har-Peled, and P.~Indyk.
\newblock Approximate clustering via core-sets.
\newblock In {\em STOC '02: Proc. of the 34th annual ACM symposium on Theory of
  computing}, pages 250-- 257, New York, NY, USA, 2002. ACM.

\bibitem{bohnenblust-38}
H.F. Bohnenblust.
\newblock Convex regions and projections in {M}inkowski spaces.
\newblock {\em Ann. Math.}, 39:301--308, 1938.

\bibitem{excursionsIntoCombGeom}
V.~Boltyanski, H.~Martini, and P.S. Soltan.
\newblock {\em Excursions into Combinatorial Geometry}.
\newblock Universitext, 1997.

\bibitem{bonnesenFenchelT}
T.~Bonnesen and W.~Fenchel.
\newblock {\em Theorie der konvexen {K}{\"o}rper}.
\newblock Springer, Berlin, 1934.
\newblock Translation: \it Theory of convex bodies, \rm BCS Associates, Moscow,
  Idaho (USA), 1987.

\bibitem{br-07}
R.~Brandenberg and L.~Roth.
\newblock New algorithms for $k$-center and extensions.
\newblock {\em Journal of Combinatorial Optimization}, 18:376--392, 2009.

\bibitem{brandenbergRoth-10}
R.~Brandenberg and L.~Roth.
\newblock Minimal containment under homothetics: a simple cutting plane
  approach.
\newblock {\em Computational Optimization and Applications}, 48:325--340, 2011.

\bibitem{swanepoel-brandenberg-private-06}
R.~Brandenberg and K.~Swanepoel.
\newblock Private communication.
\newblock 2006.

\bibitem{clarkson}
K.~L. Clarkson.
\newblock Coresets, sparse greedy approximation, and the {F}rank-{W}olfe
  algorithm.
\newblock In {\em SODA '08: Proceedings of the 19th annual ACM-SIAM symposium
  on Discrete algorithms}, pages 922--931. Society for Industrial and Applied
  Mathematics, 2008.

\bibitem{dgk-63}
L.~Danzer, B.~Gr\"unbaum, and V.~Klee.
\newblock Helly's theorem and its relatives.
\newblock In V.~Klee, editor, {\em Convexity, Proc. Symp. Pure Math., Vol. 13},
  pages 101--180. American Mathematical Society, 1963.

\bibitem{halbraumlemma}
A.~Goel, P.~Indyk, and K.~Varadarajan.
\newblock Reductions among high dimensional proximity problems.
\newblock In {\em SODA '01: Proceedings of the twelfth annual ACM-SIAM
  symposium on Discrete algorithms}, pages 769--778. Society for Industrial and
  Applied Mathematics, 2001.

\bibitem{gritzmannKlee-92}
P.~Gritzmann and V.~Klee.
\newblock Inner and outer j-radii of convex bodies in finite-dimensional normed
  spaces.
\newblock {\em Discrete Comput. Geom. 7}, pages 255--280, 1992.

\bibitem{gritzmannKlee-93}
P.~Gritzmann and V.~Klee.
\newblock Computational complexity of inner and outer j-radii of polytopes in
  finite-dimensional normed spaces.
\newblock {\em Mathematical Programming}, 59:163--213, 1993.

\bibitem{gritzmannKlee-94}
P.~Gritzmann and V.~Klee.
\newblock On the complexity of some basic problems in computational convexity.
  {I}. {C}ontainment problems.
\newblock {\em Discrete Math.}, 136:129--174, 1994.

\bibitem{gruenbaum-63}
B.~Gr{\"u}nbaum.
\newblock Measures of symmetry for convex sets.
\newblock {\em Proc. Symp. Pure Math., 1 (Convexity)}, pages 271--284, 1963.

\bibitem{henk92}
M.~Henk.
\newblock A generalization of {J}ung's theorem.
\newblock {\em Geom. Dedicata 42}, pages 235--240, 1992.

\bibitem{jung-01}
H.W.E. Jung.
\newblock {\"U}ber die kleinste {K}ugel, die eine r\"aumliche {F}igur
  einschlie{\ss}t.
\newblock {\em J. Reine Angew. Math.}, 123:241--257, 1901.

\bibitem{GLP2}
J.~Matousek M.~Sharir, E.~Welzl.
\newblock A subexponential bound for linear programming.
\newblock {\em Algorithmica, Volume 16, Numbers 4-5}, pages 498--516, April
  1996.

\bibitem{panigrahy}
R.~Panigrahy.
\newblock Minimum enclosing polytope in high dimensions.
\newblock {\em CoRR}, cs.CG/0407020, 2004.

\bibitem{pukhov-80}
S.V. Pukhov.
\newblock Kolmogorov diameters of a regular simplex.
\newblock {\em Mosc. Univ. Math. Bull.}, 35:38--41, 1980.

\bibitem{rockafellar}
R.T. Rockafellar.
\newblock {\em {Convex Analysis (Princeton Landmarks in Mathematics and
  Physics)}}.
\newblock {Princeton University Press}, December 1996.

\bibitem{svz}
A.~Saha, S.V.N. Vishwanathan, and X.~Zhang.
\newblock New approximation algorithms for minimum enclosing convex shapes.
\newblock In {\em SODA '11}. Society for Industrial and Applied Mathematics,
  2011.

\bibitem{schneider}
R.~Schneider.
\newblock {\em Convex bodies : the Brunn-Minkowski theory}.
\newblock Cambridge University Press, Cambridge, New York, 1993.

\bibitem{GLP}
M.~Sharir and E.~Welzl.
\newblock A combinatorial bound for linear programming and related problems.
\newblock In {\em STACS '92: Proceedings of the 9th Annual Symposium on
  Theoretical Aspects of Computer Science}, pages 569--579, London, UK, 1992.
  Springer-Verlag.

\bibitem{yapv-07}
H.~Yu, P.K. Agarval, and K.R.~Varadarajan R.~Poreddy.
\newblock Practical methods for shape fitting and kinetic data structures using
  coresets.
\newblock {\em Algorithmica, Volume 52, Number 3}, pages 378--402, 2007.

\end{thebibliography}

\end{document}